\documentclass[leqno, a4paper]{article}
\usepackage[left=2cm,right=2cm, top=2.5cm,bottom=2.5cm,bindingoffset=0cm]{geometry}
\usepackage{tikz, titlesec}
\usetikzlibrary{calc,intersections,through,backgrounds}
\usepackage{amsmath, amssymb, amsthm}
\usepackage{mathtools}

\mathtoolsset{showonlyrefs}

\usepackage{array, hyperref}

\titleformat{\section}{\large\bfseries\filcenter}{\thesection}{1em}{}
\titleformat{\subsection}{\bfseries}{\thesubsection}{1em}{}
\titleformat{\subsubsection}[runin]{\bfseries}{\thesubsubsection}{1em}{}[.]

\usetikzlibrary{arrows} 
\tikzset{
    >=stealth',
    pil/.style={
           ->,
           thick,
           shorten <=2pt,
           shorten >=2pt,}
}

\newcommand{\Z}{\mathbb{Z}}
\newcommand{\C}{\mathbb{C}}

\renewcommand{\P}{\mathbb{P}}

\newcommand{\Id}{\mathrm{Id}}
\newcommand{\tr}{\mathrm{trace}}

\newtheorem{lemma}{Lemma}[section]

\newtheorem{theorem}[lemma]{Theorem}
\newtheorem{corollary}[lemma]{Corollary}

\theoremstyle{definition}
\newtheorem{remark}[lemma]{Remark}
\newtheorem{definition}[lemma]{Definition}

\title{Pentagrams, inscribed polygons, and Prym varieties}    
\author{Anton Izosimov\thanks{Department of Mathematics, University of Toronto, e-mail: \tt{izosimov@math.utoronto.ca}}}
\date{}

\begin{document}
\maketitle
\abstract{The pentagram map is a discrete integrable system on the moduli space of planar polygons. The corresponding first integrals are so-called monodromy invariants $E_1, O_1, E_2, O_2,\dots$ By analyzing the combinatorics of these invariants,  R.\,Schwartz and S.\,Tabachnikov have recently proved that for polygons inscribed in a conic section one has $E_k = O_k$ for all $k$. In this paper we give a simple conceptual proof of the Schwartz-Tabachnikov theorem. Our main observation is that for inscribed polygons the corresponding monodromy satisfies a certain self-duality relation. From this we also deduce that the space of inscribed polygons with fixed values of the monodromy invariants is an open dense subset in the Prym variety (i.e., a half-dimensional torus in the Jacobian) of the spectral curve. As a byproduct, we also prove another conjecture of Schwartz and Tabachnikov on positivity of monodromy invariants for convex polygons.

}

\section{Introduction and main results}

The pentagram map was introduced by R.\,Schwartz \cite{schwartz1992pentagram} in 1992, and is now one of the most renowned discrete integrable systems which has deep connections with many different subjects such as projective geometry, integrable PDEs, cluster algebras, etc.  
The definition of the pentagram map is illustrated in Figure 1: the image of the polygon $P$ under the pentagram map is the polygon $P'$ whose vertices are the intersection points of consecutive ``short'' diagonals of~$P$ (i.e.,  diagonals connecting second-nearest vertices). 
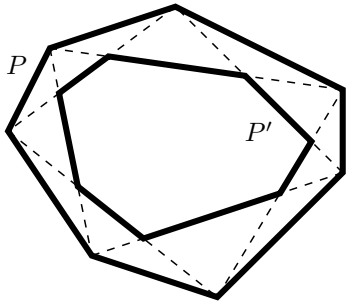
\begin{figure}[h]
\centering
\begin{tikzpicture}[thick, scale = 1.1]
\coordinate (A) at (0,0);
\coordinate (B) at (1.5,-0.5);
\coordinate (C) at (3,1);
\coordinate (D) at (3,2);
\coordinate (E) at (1,3);
\coordinate (F) at (-0.5,2.5);
\coordinate (G) at (-1,1.5);

\draw  [line width=0.8mm]  (A) -- (B) -- (C) -- (D) -- (E) -- (F) -- (G) -- cycle;
\draw [dashed, line width=0.2mm, name path=AC] (A) -- (C);
\draw [dashed,line width=0.2mm, name path=BD] (B) -- (D);
\draw [dashed,line width=0.2mm, name path=CE] (C) -- (E);
\draw [dashed,line width=0.2mm, name path=DF] (D) -- (F);
\draw [dashed,line width=0.2mm, name path=EG] (E) -- (G);
\draw [dashed,line width=0.2mm, name path=FA] (F) -- (A);
\draw [dashed,line width=0.2mm, name path=GB] (G) -- (B);

\path [name intersections={of=AC and BD,by=Bp}];
\path [name intersections={of=BD and CE,by=Cp}];
\path [name intersections={of=CE and DF,by=Dp}];
\path [name intersections={of=DF and EG,by=Ep}];
\path [name intersections={of=EG and FA,by=Fp}];
\path [name intersections={of=FA and GB,by=Gp}];
\path [name intersections={of=GB and AC,by=Ap}];

\draw  [line width=0.8mm]  (Ap) -- (Bp) -- (Cp) -- (Dp) -- (Ep) -- (Fp) -- (Gp) -- cycle;

\node at (-0.9,2.3) () {$P$};
\node at (2,1.5) () {$P'$};

\end{tikzpicture}
\caption{The pentagram map.}
\end{figure}

Since this construction is projectively invariant, one usually regards the pentagram map as a dynamical system on the space of polygons in $\P^2$ modulo projective equivalence. (Here $\P^2$ denotes the real or complex projective plane. More generally, one can consider polygons in the projective plane over any field.) %, and this is what the term \textit{pentagram map} means throughout the paper.
 The pentagram map also naturally extends to a bigger space of so-called \textit{twisted polygons}. A bi-infinite sequence of points $v_i \in \P^2$ is called a \textit{twisted $n$-gon} if $v_{i+n} = M(v_i)$ for every $i \in \Z$ and a fixed projective transformation $M$, called the \textit{monodromy}. The case~$M = \Id$ corresponds to closed polygons. The pentagram map is well-defined on the space of projective equivalence classes of twisted polygons and preserves the conjugacy class of the monodromy. \par

From the beginning there was a strong indication that the pentagram map is integrable. In \cite{schwartz2001pentagram} Schwartz established a first result in this direction proving that the pentagram map is recurrent. Further, in \cite{schwartz2008discrete} he constructed two sequences $E_k, O_k$ 
of so-called \textit{monodromy invariants} preserved by the pentagram map. 
The functions $E_k, O_k$ are, roughly speaking, weighted homogeneous components of spectral invariants of the monodromy matrix. Remarkably, this construction, essentially based on the notion of a twisted polygon, provides invariants for the pentagram map in both twisted and closed cases.\par

V.\,Ovsienko, R.\,Schwartz, and S.\,Tabachnikov \cite{ovsienko2010pentagram} proved that the pentagram map on twisted polygons has an invariant Poisson bracket, and that the monodromy invariants Poisson commute, thus establishing Arnold-Liouville integrability in the twisted case. F.\,Soloviev~\cite{soloviev2013integrability} showed that the pentagram map is algebraically integrable, both in twisted and closed cases. 
An alternative proof of integrability in the closed case can be found in \cite{ovsienko2013liouville}.\par

We also mention, in random order, several works which generalize the pentagram map and explore its relations to other subjects. M.\,Glick \cite{GLICK20111019} interpreted the pentagram map in terms of cluster algebras. M.\,Gekhtman, M.\,Shapiro, S.\,Tabachnikov, and  A.\,Vainshtein \cite{Gekhtman2016} generalized Glick's work by including the pentagram map into a family of discrete integrable systems related to weighted directed networks. In this family one finds a discrete version of the relativistic Toda lattice, as well as certain multidimensional generalizations of the pentagram map defined on so-called \textit{corrugated} polygons. Other integrable generalizations of the pentagram map were studied by B.\,Khesin and F.\,Soloviev~\cite{khesin2013, khesin2016}, G.\,Mar{\'\i} Beffa~\cite{beffa2015}, and R.\,Felipe with G.\,Mar{\'\i} Beffa~\cite{felipe2015}. Some of these maps have been recently put in the context of cluster algebra by M.\,Glick and P.\,Pylyavskyy \cite{glick2015}. We finally mention the work of V.\,Fock and A.\,Marshakov \cite{fock2014loop}, which, in particular, relates the pentagram map to Poisson-Lie groups, and the paper \cite{kedem2015t} by R.\,Kedem and P.\,Vichitkunakorn, which interprets the pentagram map in terms of T-systems.\par
\smallskip

In the present paper, we study the interaction of the pentagram map with polygons inscribed in conic sections. Schwartz and Tabachnikov \cite{schwartz2011pentagram} proved that the restrictions of the monodromy invariants to inscribed polygons satisfy the following identities.

\begin{theorem}\label{thm1} 
For polygons inscribed in a nondegenerate conic, one has $E_k = O_k$ for every $k$.
\end{theorem}
The proof of Schwartz and Tabachnikov is rather hard and is based on combinatorial analysis of explicit formulas for the monodromy invariants. Our first result is a new proof of Theorem \ref{thm1}. Our argument does not rely on explicit formulas but employs the definition of~$E_k, O_k$ in terms of the spectrum of the monodromy. Namely, we show that, up to conjugation, the monodromy matrix for inscribed polygons satisfies a \textit{self-duality relation}
\begin{align}\label{monodromyRelation}
M(z)^{-1} = M(z^{-1})^{t},
\end{align}
where $z$ is the spectral parameter. 
As a corollary, the spectral curve, defined, roughly speaking, as the zero locus of the characteristic polynomial $  \det(M(z) - w\Id)$, 
is invariant under the involution $\sigma \colon (z,w) \leftrightarrow (z^{-1}, w^{-1}),$ which implies $E_k = O_k$ for every~$k$.
\begin{remark}
Note that Theorem \ref{thm1} is, in general, not true for polygons inscribed in degenerate conics. Although any degenerate conic $C$ can be approximated by nondegenerate conics $C_\varepsilon \to C$, a twisted polygon inscribed in $C$ cannot be, in general, approximated by twisted polygons inscribed in $C_\varepsilon$. For this reason, one cannot apply a limiting argument to conclude that $E_k = O_k$ in the degenerate case, and, in fact, there are twisted polygons inscribed in degenerate conics with $E_k \neq O_k$. Nevertheless, for closed polygons, Theorem~\ref{thm1} is true in both nondegenerate and degenerate cases.

\end{remark}

We also obtain a geometric characterization of inscribed polygons with fixed values of the monodromy invariants and describe the behavior of this set under the pentagram map. For simplicity, consider the case of $n$-gons with odd $n$. In this case, the variety of polygons with fixed generic monodromy invariants $E_k, O_k$ is identified with an open dense subset in the Jacobian of the spectral curve, while the pentagram map is a shift relative to the group structure on the Jacobian~\cite{soloviev2013integrability}. Consider a level set of the monodromy invariants which contains at least one inscribed polygon. In this case, there is an involution $\sigma \colon (z,w) \leftrightarrow (z^{-1}, w^{-1})$ on the spectral curve $X$, defining a double covering $\pi \colon X \to Y := X / \sigma$. Now recall that, given a ramified (non-\'etale) double covering of curves $\pi \colon X \to Y$, one can decompose the Jacobian of $X$ in a sum of Abelian subvarieties of complementary dimensions:
\begin{align*}
\mathrm{Jac}(X) = \mathrm{Jac}(Y) + \mathrm{Prym}(X \vert Y).
\end{align*}
Here $\mathrm{Jac}(Y)$ is the Jacobian of the base curve $Y$ embedded in $\mathrm{Jac}(X)$ by means of the pullback homomorphism $\pi^*$, while $\mathrm{Prym}(X \vert Y)$ is the \textit{Prym variety} of $X$ over $Y$, which can be defined as the kernel of the pushforward (or \textit{norm}) homomorphism $\pi_*$.
The intersection $\mathrm{Jac}(Y) \cap \mathrm{Prym}(X \vert Y)$ is the finite set of order~$2$ points in $\mathrm{Jac}(Y)$.\par

 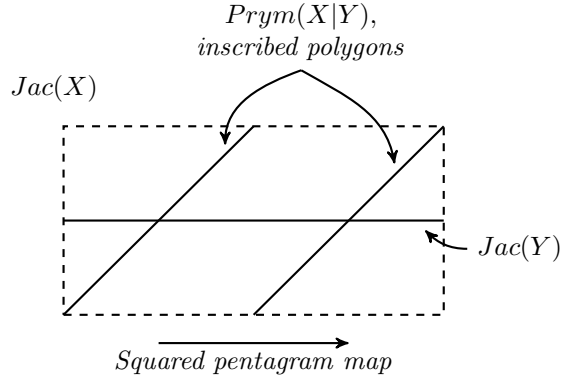
\begin{figure}[t]
\centerline{
\qquad
\begin{tikzpicture}[thick, scale = 1.25, every text node part/.style={align=center},]
\draw [dashed] (0,0) -- (0,2)-- (4,2) -- (4,0) -- cycle;
\draw [](0,1) -- (4,1);
\draw [](0,0) -- (2,2);
\draw [](2,0) -- (4,2);
  \node[text width=5cm] at (2.5,3) (A) { ${Prym}(X\vert Y)$, \\ \textit{inscribed polygons}}; 
  \node at (1.7,1.7) (B) {};
 \draw [->] (A.south) to [out = 200, in=90] (B);
  \node at (3.5,1.5) (C) {};
  \draw [->] (A.south) to [out = 330, in=100] (C);
   \node at (4.8, 0.7) (D) { ${Jac}(Y)$}; 
   \node at (3.7, 1) (E) {};
    \draw [->] (D) to [out = 180, in=315] (E);
    \draw [->] (1,-0.3) -- (3,-0.3);
       \node at (2,-0.5) () { \textit{Squared pentagram map}}; 
          \node at (-0.1, 2.4) () { ${Jac}(X)$}; 
\end{tikzpicture}
}
\caption{The pentagram map and inscribed polygons.}\label{pmip}
\end{figure}
In the following theorem, we consider twisted polygons in the complex projective plane.
\begin{theorem}\label{thm2}
\begin{enumerate}
\item
The space of twisted inscribed $(2q+1)$-gons with fixed generic monodromy invariants $E_k = O_k$ is an open dense subset in a subtorus of $\mathrm{Jac}(X)$ parallel\footnote{In fact, since the identification between the level set of the monodromy invariants and $\mathrm{Jac}(X)$ is only defined up to a shift, one can assume that this subtorus is exactly $\mathrm{Prym}(X \vert Y)$.} to the Prym variety~$\mathrm{Prym}(X \vert Y)$.
\item The square of the pentagram map is, up to a relabeling of vertices, a translation along the complementary Abelian subvariety $ \mathrm{Jac}(Y) $.
\end{enumerate}
\end{theorem}
See Figure \ref{pmip} (the opposite sides of the rectangle are glued together to make a torus).
\begin{remark}
Let us comment on the dimensions of the sets mentioned in Theorem \ref{thm2}. For twisted~$(2q+1)$-gons, the number of monodromy invariants is $2(q+1)$. The restriction $E_k = O_k$ gives a $(q+1)$-dimensional space of possible values of invariants. Theorem \ref{thm2} says that the variety of inscribed polygons modulo projective transformations is a fibration over this $(q+1)$-dimensional space with generic fiber $\mathrm{Prym}(X \vert Y)$. The genus of the spectral curve $X$ is generically $2q$. The involution~$(z,w) \leftrightarrow (z^{-1}, w^{-1})$ on $X$ has two fixed points, therefore, by the Riemann-Hurwitz formula, the genus of $Y$ is $q$. So, $\dim \mathrm{Jac}(Y) = q$, and $\dim \mathrm{Prym}(X \vert Y) = \dim \mathrm{Jac}(X) - \dim \mathrm{Jac}(Y) = q$. \par We conclude that the moduli space of inscribed twisted $(2q+1)$-gons is a fibration over a $(q+1)$-dimensional base with $q$-dimensional fibers, and thus has total dimension $2q+1$. Note that the latter dimension can also be computed by identifying inscribed polygons with polygons in $\P^1$.
\end{remark}
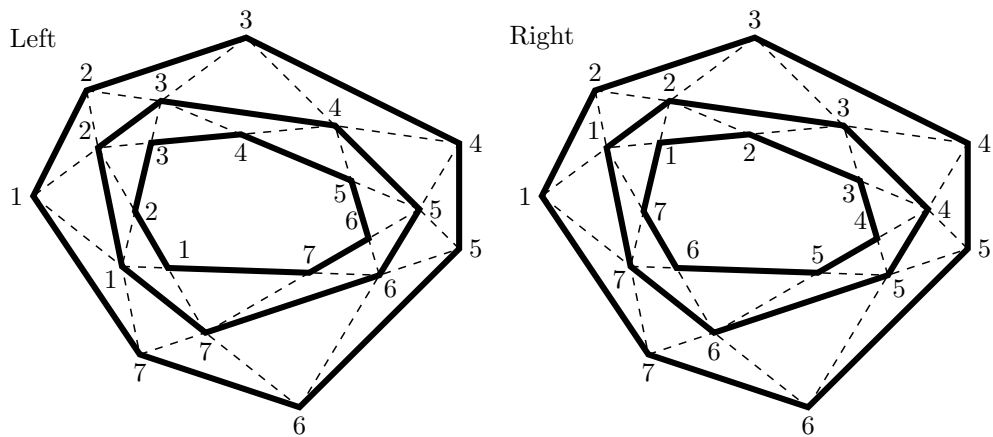
\begin{figure}[b]
\centering
\begin{tikzpicture}[thick, scale = 1.4]
\node at (-1,3) () {Left};
\coordinate[label={below:$7$}] (A) at (0,0);
\coordinate [label={below:$6$}](B) at (1.5,-0.5);
\coordinate [label={right:$5$}](C) at (3,1);
\coordinate [label={right:$4$}](D) at (3,2);
\coordinate [label={$3$}](E) at (1,3);
\coordinate [label={$2$}](F) at (-0.5,2.5);
\coordinate [label={left:$1$}](G) at (-1,1.5);

\draw  [line width=0.8mm]  (A) -- (B) -- (C) -- (D) -- (E) -- (F) -- (G) -- cycle;
\draw [dashed, line width=0.2mm, name path=AC] (A) -- (C);
\draw [dashed,line width=0.2mm, name path=BD] (B) -- (D);
\draw [dashed,line width=0.2mm, name path=CE] (C) -- (E);
\draw [dashed,line width=0.2mm, name path=DF] (D) -- (F);
\draw [dashed,line width=0.2mm, name path=EG] (E) -- (G);
\draw [dashed,line width=0.2mm, name path=FA] (F) -- (A);
\draw [dashed,line width=0.2mm, name path=GB] (G) -- (B);

\path [name intersections={of=AC and BD,by=Bp}];
\path [name intersections={of=BD and CE,by=Cp}];
\path [name intersections={of=CE and DF,by=Dp}];
\path [name intersections={of=DF and EG,by=Ep}];
\path [name intersections={of=EG and FA,by=Fp}];
\path [name intersections={of=FA and GB,by=Gp}];
\path [name intersections={of=GB and AC,by=Ap}];

\draw [dashed, line width=0.2mm, name path=ACp] (Ap) -- (Cp);
\draw [dashed,line width=0.2mm, name path=BDp] (Bp) -- (Dp);
\draw [dashed,line width=0.2mm, name path=CEp] (Cp) -- (Ep);
\draw [dashed,line width=0.2mm, name path=DFp] (Dp) -- (Fp);
\draw [dashed,line width=0.2mm, name path=EGp] (Ep) -- (Gp);
\draw [dashed,line width=0.2mm, name path=FAp] (Fp) -- (Ap);
\draw [dashed,line width=0.2mm, name path=GBp] (Gp) -- (Bp);

\path [name intersections={of=ACp and BDp,by=Bpp}];
\path [name intersections={of=BDp and CEp,by=Cpp}];
\path [name intersections={of=CEp and DFp,by=Dpp}];
\path [name intersections={of=DFp and EGp,by=Epp}];
\path [name intersections={of=EGp and FAp,by=Fpp}];
\path [name intersections={of=FAp and GBp,by=Gpp}];
\path [name intersections={of=GBp and ACp,by=App}];

\coordinate[label={below:$7$}] (appp) at (Ap);
\coordinate [label={[label distance = -0.1cm]-45:$6$}](bppp) at (Bp);
\coordinate [label={right:$5$}](cppp) at (Cp);
\coordinate [label={$4$}](dppp) at (Dp);
\coordinate [label=$3$](eppp) at (Ep);
\coordinate [label={[label distance = -0.05cm]170:{$2$}}](fppp) at (Fp);
\coordinate [label={[label distance = -0.1cm]225:{$1$}}](gppp) at (Gp);

\coordinate[label={$7$}] (appp) at (App);
\coordinate [label={175:$6$}](bppp) at (Bpp);
\coordinate [label={[label distance = -0.1cm]-135:$5$}](cppp) at (Cpp);
\coordinate [label={below:$4$}](dppp) at (Dpp);
\coordinate [label={[label distance = -0.1cm]-45:$3$}](eppp) at (Epp);
\coordinate [label={right:{$2$}}](fppp) at (Fpp);
\coordinate [label={45:$1$}](gppp) at (Gpp);

\draw  [line width=0.8mm]  (Ap) -- (Bp) -- (Cp) -- (Dp) -- (Ep) -- (Fp) -- (Gp) -- cycle;
\draw  [line width=0.8mm]  (App) -- (Bpp) -- (Cpp) -- (Dpp) -- (Epp) -- (Fpp) -- (Gpp) -- cycle;

\end{tikzpicture}
\begin{tikzpicture}[thick, scale = 1.4]
\node at (-1,3) () {Right};
\coordinate[label={below:$7$}] (A) at (0,0);
\coordinate [label={below:$6$}](B) at (1.5,-0.5);
\coordinate [label={right:$5$}](C) at (3,1);
\coordinate [label={right:$4$}](D) at (3,2);
\coordinate [label={$3$}](E) at (1,3);
\coordinate [label={$2$}](F) at (-0.5,2.5);
\coordinate [label={left:$1$}](G) at (-1,1.5);

\draw  [line width=0.8mm]  (A) -- (B) -- (C) -- (D) -- (E) -- (F) -- (G) -- cycle;
\draw [dashed, line width=0.2mm, name path=AC] (A) -- (C);
\draw [dashed,line width=0.2mm, name path=BD] (B) -- (D);
\draw [dashed,line width=0.2mm, name path=CE] (C) -- (E);
\draw [dashed,line width=0.2mm, name path=DF] (D) -- (F);
\draw [dashed,line width=0.2mm, name path=EG] (E) -- (G);
\draw [dashed,line width=0.2mm, name path=FA] (F) -- (A);
\draw [dashed,line width=0.2mm, name path=GB] (G) -- (B);

\path [name intersections={of=AC and BD,by=Bp}];
\path [name intersections={of=BD and CE,by=Cp}];
\path [name intersections={of=CE and DF,by=Dp}];
\path [name intersections={of=DF and EG,by=Ep}];
\path [name intersections={of=EG and FA,by=Fp}];
\path [name intersections={of=FA and GB,by=Gp}];
\path [name intersections={of=GB and AC,by=Ap}];

\draw [dashed, line width=0.2mm, name path=ACp] (Ap) -- (Cp);
\draw [dashed,line width=0.2mm, name path=BDp] (Bp) -- (Dp);
\draw [dashed,line width=0.2mm, name path=CEp] (Cp) -- (Ep);
\draw [dashed,line width=0.2mm, name path=DFp] (Dp) -- (Fp);
\draw [dashed,line width=0.2mm, name path=EGp] (Ep) -- (Gp);
\draw [dashed,line width=0.2mm, name path=FAp] (Fp) -- (Ap);
\draw [dashed,line width=0.2mm, name path=GBp] (Gp) -- (Bp);

\path [name intersections={of=ACp and BDp,by=Bpp}];
\path [name intersections={of=BDp and CEp,by=Cpp}];
\path [name intersections={of=CEp and DFp,by=Dpp}];
\path [name intersections={of=DFp and EGp,by=Epp}];
\path [name intersections={of=EGp and FAp,by=Fpp}];
\path [name intersections={of=FAp and GBp,by=Gpp}];
\path [name intersections={of=GBp and ACp,by=App}];

\coordinate[label={below:$6$}] (appp) at (Ap);
\coordinate [label={[label distance = -0.1cm]-45:$5$}](bppp) at (Bp);
\coordinate [label={right:$4$}](cppp) at (Cp);
\coordinate [label={$3$}](dppp) at (Dp);
\coordinate [label=$2$](eppp) at (Ep);
\coordinate [label={[label distance = -0.05cm]170:{$1$}}](fppp) at (Fp);
\coordinate [label={[label distance = -0.1cm]225:{$7$}}](gppp) at (Gp);

\coordinate[label={$5$}] (appp) at (App);
\coordinate [label={175:$4$}](bppp) at (Bpp);
\coordinate [label={[label distance = -0.1cm]-135:$3$}](cppp) at (Cpp);
\coordinate [label={below:$2$}](dppp) at (Dpp);
\coordinate [label={[label distance = -0.1cm]-45:$1$}](eppp) at (Epp);
\coordinate [label={right:{$7$}}](fppp) at (Fpp);
\coordinate [label={45:$6$}](gppp) at (Gpp);

\draw  [line width=0.8mm]  (Ap) -- (Bp) -- (Cp) -- (Dp) -- (Ep) -- (Fp) -- (Gp) -- cycle;
\draw  [line width=0.8mm]  (App) -- (Bpp) -- (Cpp) -- (Dpp) -- (Epp) -- (Fpp) -- (Gpp) -- cycle;

\end{tikzpicture}

\caption{Left and right labeling schemes.}\label{LR}
\end{figure}

\begin{remark}
The reason why in part 2 of Theorem \ref{thm2} one needs to relabel the vertices to identify the squared pentagram map with a translation along $ \mathrm{Jac}(Y) $ is the following. There are two equally natural ways (the so-called \textit{left and right labeling schemes}) how to define the pentagram map on polygons with labeled vertices. Figure \ref{LR} depicts these two labeling schemes and the squares of the corresponding maps. Notice that in both cases, the squared map exhibits a shift of indices. At the same time, the translation along~$ \mathrm{Jac}(Y) $ corresponds to the ``canonical squared pentagram map'' depicted in Figure \ref{SQ} (the latter is, in fact, not a square of any map).

\end{remark}

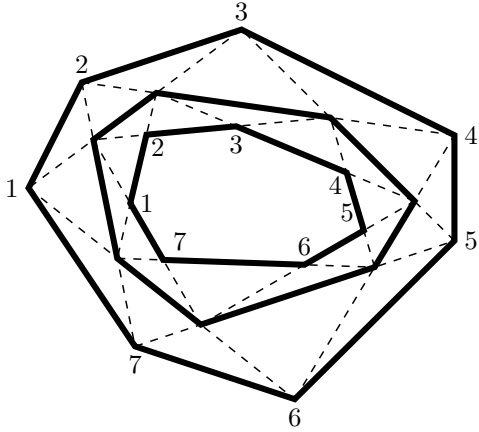
\begin{figure}[t]
\centering
\begin{tikzpicture}[thick, scale = 1.4]
\coordinate[label={below:$7$}] (A) at (0,0);
\coordinate [label={below:$6$}](B) at (1.5,-0.5);
\coordinate [label={right:$5$}](C) at (3,1);
\coordinate [label={right:$4$}](D) at (3,2);
\coordinate [label={$3$}](E) at (1,3);
\coordinate [label={$2$}](F) at (-0.5,2.5);
\coordinate [label={left:$1$}](G) at (-1,1.5);

\draw  [line width=0.8mm]  (A) -- (B) -- (C) -- (D) -- (E) -- (F) -- (G) -- cycle;
\draw [dashed, line width=0.2mm, name path=AC] (A) -- (C);
\draw [dashed,line width=0.2mm, name path=BD] (B) -- (D);
\draw [dashed,line width=0.2mm, name path=CE] (C) -- (E);
\draw [dashed,line width=0.2mm, name path=DF] (D) -- (F);
\draw [dashed,line width=0.2mm, name path=EG] (E) -- (G);
\draw [dashed,line width=0.2mm, name path=FA] (F) -- (A);
\draw [dashed,line width=0.2mm, name path=GB] (G) -- (B);

\path [name intersections={of=AC and BD,by=Bp}];
\path [name intersections={of=BD and CE,by=Cp}];
\path [name intersections={of=CE and DF,by=Dp}];
\path [name intersections={of=DF and EG,by=Ep}];
\path [name intersections={of=EG and FA,by=Fp}];
\path [name intersections={of=FA and GB,by=Gp}];
\path [name intersections={of=GB and AC,by=Ap}];

\draw [dashed, line width=0.2mm, name path=ACp] (Ap) -- (Cp);
\draw [dashed,line width=0.2mm, name path=BDp] (Bp) -- (Dp);
\draw [dashed,line width=0.2mm, name path=CEp] (Cp) -- (Ep);
\draw [dashed,line width=0.2mm, name path=DFp] (Dp) -- (Fp);
\draw [dashed,line width=0.2mm, name path=EGp] (Ep) -- (Gp);
\draw [dashed,line width=0.2mm, name path=FAp] (Fp) -- (Ap);
\draw [dashed,line width=0.2mm, name path=GBp] (Gp) -- (Bp);

\path [name intersections={of=ACp and BDp,by=Bpp}];
\path [name intersections={of=BDp and CEp,by=Cpp}];
\path [name intersections={of=CEp and DFp,by=Dpp}];
\path [name intersections={of=DFp and EGp,by=Epp}];
\path [name intersections={of=EGp and FAp,by=Fpp}];
\path [name intersections={of=FAp and GBp,by=Gpp}];
\path [name intersections={of=GBp and ACp,by=App}];
\coordinate[label={$6$}] (appp) at (App);
\coordinate [label={175:$5$}](bppp) at (Bpp);
\coordinate [label={[label distance = -0.1cm]-135:$4$}](cppp) at (Cpp);
\coordinate [label={below:$3$}](dppp) at (Dpp);
\coordinate [label={[label distance = -0.1cm]-45:$2$}](eppp) at (Epp);
\coordinate [label={right:{$1$}}](fppp) at (Fpp);
\coordinate [label={45:$7$}](gppp) at (Gpp);

\draw  [line width=0.8mm]  (Ap) -- (Bp) -- (Cp) -- (Dp) -- (Ep) -- (Fp) -- (Gp) -- cycle;
\draw  [line width=0.8mm]  (App) -- (Bpp) -- (Cpp) -- (Dpp) -- (Epp) -- (Fpp) -- (Gpp) -- cycle;

\end{tikzpicture}

\caption{The natural labeling for the squared pentagram map.}\label{SQ}
\end{figure}

The proof of Theorem \ref{thm2} is somewhat involved and is based on relation \eqref{monodromyRelation}. Details of the proof will be published elsewhere.
\par
The structure of the paper is as follows. In Section \ref{sec2}, we recall the derivation of the monodromy invariants. In Section \ref{sec25}, we show that the monodromy invariants arise as coefficients in the equation of the spectral curve. In Section \ref{sec3}, we derive a formula for the monodromy in terms of the corner invariants. Finally, in Section~\ref{sec4}, we use these results to prove Theorem \ref{thm1}. Section~\ref{secLast} is devoted to open questions. Additionally, the paper contains an appendix where we apply our technique to prove another conjecture of Schwartz and Tabachnikov on positivity of monodromy invariants for convex polygons.

To conclude this introduction, we also mention two works \cite{schwartz2010elementary, schwartz2015pentagram} both studying the pentagram map or its variations for polygons inscribed or circumscribed about a conic section. We find it an interesting problem to obtain an algebraic geometric explanation of these results in the spirit of the present paper.
\medskip

{\bf Acknowledgments.} The author is grateful to Michael Gekhtman, Boris Khesin, Gloria Mar{\'\i} Beffa, Richard Schwartz, and Sergei Tabachnikov for fruitful discussions. 

\bigskip

\section{Monodromy invariants}\label{sec2}
In this section, we briefly recall the derivation of first integrals for the pentagram map -- the monodromy invariants $E_k, O_k$. Details can be found in almost any paper on the subject, for instance in~\cite{ovsienko2010pentagram, schwartz2008discrete}. \par
Let $\mathcal P_n$ be the space of twisted $n$-gons modulo projective transformations. Recall that a \textit{twisted $n$-gon} is a bi-infinite sequence of points $v_i \in \P^2$ such that $v_{i+n} = M(v_i)$ for every $i \in \Z$ and a fixed projective transformation $M$, called the \textit{monodromy}. For an equivalence class of polygons $P \in \mathcal P_n$, the monodromy $M  \in \P \mathrm{GL}_3$ is well-defined up to conjugation (this means that, as a matrix, the monodromy is defined up to conjugation and scalar multiplication).\par
In what follows, we only consider polygons satisfying the following genericity assumption: any three consecutive vertices $v_i, v_{i+1}, v_{i+2}$ are not collinear. Note that for inscribed polygons this always holds.\par
Any twisted $n$-gon is uniquely determined by its $n$ consecutive vertices and its monodromy~$M \in \P\mathrm{GL}_3$. Therefore, the dimension of the space of twisted $n$-gons is $2n+8$, while the dimension of the moduli space $\mathcal P_n$ is $2n$.
As coordinates on the space~$\mathcal P_n$, one can take the so-called \textit{corner invariants}. To every vertex $v_i$ of a twisted $n$-gon, one associates two cross-ratios $x_i, y_i$, as shown in Figure \ref{CI}. We note that while there are several different ways how to define the cross-ratio, the one traditionally used in the definition of corner invariants is
$$
[t_1, t_2,t_3, t_4] := \frac{(t_1 - t_2)(t_3 - t_4)}{(t_1 - t_3)(t_2 - t_4)}.
$$

Clearly, the sequences $(x_i)$, $(y_i)$ are $n$-periodic and depend only on the projective equivalence class of the polygon $(v_i)$. Furthermore, $\{ x_1, y_1, \dots, x_n, y_n\}$ is a coordinate system on an open dense subset of $\mathcal P_n$ \cite{schwartz2008discrete}. \begin{figure}[b]
\centering
\begin{tikzpicture}[thick, scale = 1.4, rotate = -90]
\draw [line width=0.5mm] (0.9, -0.3) -- (0.7,0.4) -- (1,1) -- (2,1.5) -- (3,1) -- (3.3,0.4) -- (3.1,-0.3);
\fill (0.7,0.4) circle [radius=2pt];
\coordinate [label={$v_{i-2}$}]() at (0.7, 0.4);
\fill (1,1) circle [radius=2pt];
\coordinate [label={45:$v_{i-1}$}]() at (1,1);
\fill (2,1.5) circle [radius=2pt];
\coordinate [label={left:$v_{i}$}]() at (2,1.5);
\fill (3,1) circle [radius=2pt];
\coordinate [label={-45:$v_{i+1}$}]() at (3,1);
\fill (3.3,0.4) circle [radius=2pt];
\coordinate [label={below:$v_{i+2}$}]() at (3.3,0.4);
\fill (2,3) circle [radius=2pt];
\coordinate [label={right:$\bar v_{i}$}]() at (2,3);
\fill (1.4,1.8) circle [radius=2pt];
\coordinate [label={45:$\tilde v_{i}$}]() at (1.4,1.8);
\fill (2.6,1.8) circle [radius=2pt];
\coordinate [label={-45:$\hat v_{i}$}]() at (2.6,1.8);
\draw [dashed] (0.7, 0.4) -- (2, 3);
\draw [dashed] (1.4, 1.8) -- (3,1);
\draw [dashed] (2, 3) -- (3.3,0.4);
\draw [dashed]  (1,1) -- (2.6,1.8);
\node at (1.8,5) () {$x_i := [v_{i-2}, v_{i-1}, \tilde v_i, \bar v_i]$};
\node at (2.2,5) () {$y_i := [\bar v_i, \hat v_i, v_{i+1}, v_{i+2}]$};
\end{tikzpicture}
\caption{The definition of the corner invariants.}\label{CI}
\end{figure}
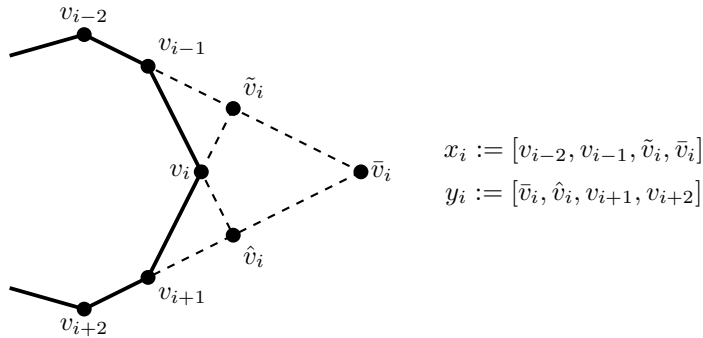

To define the pentagram map, we use the right labeling scheme (see Figure \ref{LR}). The image of a polygon~$(v_i)$ is the polygon~$(v_i')$ such that $v_i' = [v_i, v_{i+2}] \cap [v_{i-1}, v_{i+1}]$. This map descends to a densely defined map $T \colon \mathcal P_n \to \mathcal P_n$ (which we also call the pentagram map). In terms of the corner invariants $x_i, y_i$, the map $T$ is given by
\begin{align}\label{pentFormulas}
x_i' = x_i \frac{1 - x_{i-1}y_{i-1}}{1 - x_{i+1}y_{i+1}},\quad
y_i' = y_{i+1} \frac{1 - x_{i+2}y_{i+2}}{1 - x_{i}y_{i}}.
\end{align}
Now, we recall the construction of the monodromy invariants. For any equivalence class $P \in \mathcal P_n$ of twisted $n$-gons, the corresponding monodromy $M$ is a $3 \times 3$  matrix defined up to scalar multiplication and conjugation. Therefore, the quantities
$$
\Omega_1 := \frac{\tr^3(M^{-1})}{\det(M^{-1})}, \quad \Omega_2 := \frac{\tr^3(M)}{\det(M)}
$$
are well-defined\footnote{There is a typo in the corresponding formula (2.8) of \cite{ovsienko2010pentagram}. To be consistent with subsequent definitions, $\Omega_1$ should be defined using $M^{-1}$, and $\Omega_2$ using $M$, not the other way around.}. Furthermore, since the pentagram map on polygons preserves the monodromy, the pentagram map  $T \colon \mathcal P_n \to \mathcal P_n$ on equivalence classes preserves $\Omega_1, \Omega_2$. 
Further, let 
$$
O_n := \prod_{i=1}^n x_i, \quad E_n := \prod_{i=1}^n y_i.
$$
Then from formulas \eqref{pentFormulas} it immediately follows that the functions $O_n, E_n$ are invariant under the pentagram map. Therefore, the quantities 
$$
\tilde \Omega_1 := O_n^2 E_n \Omega_1, \quad \tilde \Omega_2 := O_nE_n^2 \Omega_2
$$
are also invariant. In \cite{schwartz2008discrete} it is shown that the $\tilde \Omega_1, \tilde \Omega_2$ are polynomials in corner invariants~$x_i, y_i$. %(later on, we will obtain an independent proof of this fact).
\par
Another corollary of  \eqref{pentFormulas} is that the pentagram map commutes with the \textit{rescaling operation} \begin{align}\label{rescal}R_z \colon (x_1, y_1, \dots, x_n, y_n) \mapsto (zx_1, z^{-1}y_1, \dots, zx_n, z^{-1}y_n).\end{align}
Therefore, all coefficients in $z$ of the polynomials $
R_z^*(\tilde \Omega_1) $, $R_z^*(\tilde \Omega_2)
$, where $R_z^*$ denotes the natural action of $R_z$ on functions of corner invariants, are also first integrals.  In~\cite{schwartz2008discrete} it is shown that
\begin{align}\label{monoInv}
R_z^*(\tilde \Omega_1) = \left(1+\sum_{k=1}^{[n/2]}O_kz^k\right)^{\!\!3}\!, \quad  R_z^*(\tilde \Omega_2) = \left(1+\sum_{k=1}^{[n/2]}E_kz^{-k}\right)^{\!\!3}\!,
\end{align}
where $O_k, E_k$ are certain polynomials in corner invariants.
\begin{definition}
The functions $O_k, E_k$, where $k = 1, \dots, [n/2], n$, are called \textit{the monodromy invariants}.
\end{definition}
The monodromy invariants are polynomials in $x_i, y_i$ preserved by the pentagram map. In \cite{ovsienko2010pentagram} they are used to prove that the pentagram map is a completely integrable system.\par \medskip
\section{The spectral curve}\label{sec25}

In this section, we define the spectral curve and relate it to the monodromy invariants. The pentagram spectral curve first appeared in \cite{soloviev2013integrability}, where it was used to prove algebraic integrability.\par
Let $P \in \mathcal P_n$ be an equivalence class of twisted $n$-gons. Then the rescaling operation $R_z \colon \mathcal P_n \to \mathcal P_n$ defines a family $P_z \in \mathcal P_n, P_z := R_z(P)$ parametrized by $z \in \C^*$. For each equivalence class $P_z$, we have the corresponding conjugacy class of monodromies. Take any representative of this conjugacy class and lift it to $\mathrm{GL}_3$. This gives a family $M(z)$ of $3 \times 3$ matrices parametrized by $z \in \C^*$.

\begin{definition}
We call $M(z)$ a \textit{scaled monodromy matrix} of the equivalence class $P$.
\end{definition}
The scaled monodromy matrix of a given equivalence class is determined up to multiplication by a scalar function of $z$ and conjugation by a $z$-dependent matrix.
\par
 Now, given an equivalence class $P \in \mathcal P_n$, we take its scaled monodromy $M(z)$ and set
$$
R(z,w) := \frac{1}{\lambda(z)^3}\det(\lambda(z) w I - M(z)),
$$
where $\lambda(z) := (z^{n}\det M(z))^{1/3}$, and $I$ is the identity matrix. The function $R(z,w) $ is the characteristic polynomial of $M(z)$ normalized in such a way that it does not change if $M(z)$ is multiplied by a scalar function of $z$. Furthermore, $R(z,w)$ depends only on the conjugacy class of $M(z)$, thus being a well-defined function of $z$, $w$, and the equivalence class $P$. Explicitly, we have
\begin{align}\label{specPolynomial}
\begin{aligned}
R(z,w) = w^3 - E(z^{-1})w^2  + O(z)z^{-n}w - z^{-n},
\end{aligned}
\end{align}
where
$$
E(z) := E_n^{-\frac{2}{3}}O_n^{-\frac{1}{3}}\left(1+\sum_{k=1}^{[n/2]}E_kz^{k}\right)\!, \quad O(z) := O_n^{-\frac{2}{3}}E_n^{-\frac{1}{3}}\left(1+\sum_{k=1}^{[n/2]}O_kz^{k}\right)\!,
$$
and $E_k, O_k$ are the monodromy invariants of the equivalence class $P$.
\begin{definition}
The zero locus of $R(z,w)$ in the complex torus $(\C^*)^2$ is called the \textit{spectral curve}.
\end{definition}
The spectral curve encodes all the monodromy invariants $E_k$, $O_k$. In particular, we have the following immediate corollary of \eqref{specPolynomial}.
\begin{corollary}\label{mainCor}
The following conditions are equivalent.
\begin{enumerate} 
\item The spectral curve is invariant under the involution $(z,w) \leftrightarrow (z^{-1}, w^{-1})$.
\item The monodromy invariants satisfy $E_k = O_k$ for all $k = 1, \dots, [n/2], n$.
\end{enumerate}
\end{corollary}
\begin{remark}
The involution on the spectral curve can be illustrated by means of the Newton polygon. Recall that the Newton polygon of an algebraic curve $\sum a_{ij}z^iw^j = 0$ is the convex hull of $\{ (i,j) \in \Z^2 \mid a_{ij} \neq 0\}$. The Newton polygon of the pentagram spectral curve is a parallelogram depicted in Figure \ref{NP} (here we take $n=7$). The labels $E_k, O_k$ at integer points are, up to factors $E_n^{-{2}/{3}}O_n^{-{1}/{3}}$, $O_n^{-{2}/{3}}E_n^{-{1}/{3}}$, the coefficients of the corresponding monomials. The involution $(z,w) \leftrightarrow (z^{-1}, w^{-1})$ corresponds to the symmetry of the parallelogram with respect to its center $(-{n}/{2}, {3}/{2})$, interchanging $E_k$ with $O_k$.
\end{remark}

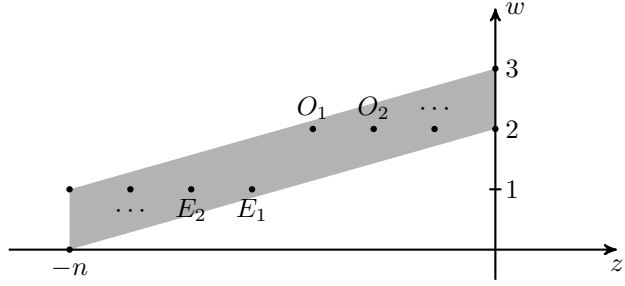
\begin{figure}[t]
\centering
\begin{tikzpicture}[thick, scale = 0.8]
\draw [->] (-8,0) -- (2,0);
\draw [->] (0,-0.5) -- (0,4);
\fill [opacity = 0.3] (0,2) -- (0,3) -- (-7,1) -- (-7,0) -- cycle;
 \fill (-7,0) circle [radius=1.5pt];
  \fill (-7,1) circle [radius=1.5pt];
   \fill (0,2) circle [radius=1.5pt];
    \fill (0,3) circle [radius=1.5pt];
      \fill (-6,1) circle [radius=1.5pt];
        \fill (-5,1) circle [radius=1.5pt];
                \fill (-4,1) circle [radius=1.5pt];
                  \fill (-3,2) circle [radius=1.5pt];
         \fill (-2,2) circle [radius=1.5pt];
                 \fill (-1,2) circle [radius=1.5pt];  
                 \coordinate [label={right:$w$}]() at (0,4);
                   \coordinate [label={below:$z$}]() at (2,0);    
                   \coordinate [label={below:$-n$}]() at (-7,0);  
                  \coordinate [label={right:$2$}]() at (0,2);       
                     \coordinate [label={right:$3$}]() at (0,3);   
                       \coordinate [label={right:$1$}]() at (0,1);     
                       \draw (-0.1,1) -- (0.1,1);  
                          \coordinate [label={below:$\vphantom{E_3}\cdots$}]() at (-6,1);
                        \coordinate [label={below:$E_2$}]() at (-5,1);
                          \coordinate [label={below:$E_1$}]() at (-4,1);
                        \coordinate [label={above:$O_2$}]() at (-2,2);
                          \coordinate [label={above:$O_1$}]() at (-3,2);     
                                   \coordinate [label={above:$\vphantom{O_3}\cdots$}]() at (-1,2);                       

\end{tikzpicture}
\caption{The Newton polygon of the spectral curve.}\label{NP}
\end{figure}
\begin{remark}
Note that in \cite{soloviev2013integrability} the spectral curve is defined using the matrix $M(z^{-1})^{-1}$. 
Although this definition is not completely equivalent to ours, the corresponding curves are related by a change of variables.

\end{remark}
\medskip
\section{The monodromy matrix via corner invariants}\label{sec3}
In this section, we express the monodromy in terms of the corner invariants. Our approach is similar to that of \cite{soloviev2013integrability}, but we do not assume that the number of vertices is not divisible by $3$. To get rid of this assumption, we use the same idea as in Remark 4.4 of \cite{ovsienko2010pentagram}.  

\begin{lemma}\label{lemma1}
Let $(v_i)$ be a twisted $n$-gon in $\P^2$ with monodromy $M \in \P\mathrm{GL}_3$ and corner invariants $(x_i, y_i)$. Then $M$ is conjugate to any of the matrices\footnote{When we say that  $M \in \P\mathrm{GL}_3$ is conjugate to a matrix $M' \in \mathrm{GL}(3)$, we mean that $M$ is conjugate to the projection of $M'$ to $ \P\mathrm{GL}_3$.}
$
 M_i :=L_iL_{i+1} \cdots L_{i+n-1},
$
where $i \in \Z$ is arbitrary, and
\begin{align}\label{li}
L_{j} := \left(\begin{array}{ccc}0 & 0 & 1 \\ -x_{j}y_{j} & 0 & 1 \\0 & -y_{j} & 1\end{array}\right).
\end{align}
\end{lemma}
\begin{proof}
 
Using that $v_{i}, v_{i+1}, v_{i+2}$ are not collinear, we lift the points $v_i$ to vectors $V_i \in \C^3$ in such a way that
\begin{align}\label{detCond}
\det (V_i, V_{i+1}, V_{i+2}) = 1 \quad \,\forall\, i \in \Z.
\end{align}
(Here and in what follows, we regard $V_i$'s as column vectors.) From \eqref{detCond} it follows that 
\begin{align}\label{diffEqn}
V_{i+3} = a_i V_{i+2} + b_i V_{i+1} + V_i
\end{align}
for some sequences $a_i, b_i \in \C$. In matrix form, this can be written as
\begin{align}\label{matDiffEqn}
W_{i+1} = W_iN_i,
\end{align}
where \begin{align}\label{ni}
W_i :=  (V_i, V_{i+1}, V_{i+2}), \quad N_i := \left(\begin{array}{ccc}0 & 0 & 1 \\1 & 0 & b_i \\0 & 1 & a_i\end{array}\right).
\end{align}
It follows that
\begin{align}\label{nSteps}
W_{i+n} = W_i (N_iN_{i+1} \cdots N_{i+n-1}).
\end{align}
Further, take an arbitrary lift $\tilde M \in \mathrm{GL}_3$ of the monodromy $M$. Note that since $  M(v_i) = v_{i+n}$, we have
$
 \tilde MV_i =t_iV_{i+n}
$ 
for some sequence $t_i \in \C$. This can be rewritten as
\begin{align}\label{matMonEqn}
  \tilde MW_i = W_{i+n}D_i,
\end{align}
where $D_i$ is a diagonal matrix $D_i := \mathrm{diag}(t_i, t_{i+1}, t_{i+2})$. Comparing \eqref{nSteps} and \eqref{matMonEqn}, we conclude that $\tilde M$ (and thus $M$) is conjugate to
\begin{align}\label{mimat}
 \tilde M_i := (N_iN_{i+1} \cdots N_{i+n-1})D_i,
\end{align}
for any $i \in \Z$.\par
Now, we use the result of  Lemma 4.5 of \cite{ovsienko2010pentagram} which says that given any lift $(V_i)$ of a twisted polygon satisfying \eqref{diffEqn}, the corner invariants are given by
$$
x_i = \frac{a_{i-2}}{b_{i-2}b_{i-1}}, \quad y_i = -\frac{b_{i-1}}{a_{i-2}a_{i-1}}.
$$
Using these formulas, one easily verifies that the matrix $L_j$ given by \eqref{li} is related to $N_j$ given by \eqref{ni} by means of a gauge type transformation:
$$
L_{j} = \frac{1}{a_{j-1}}\Lambda_{j-2}^{-1}N_{j-2}\Lambda_{j-1}, \quad \Lambda_i := \mathrm{diag}(1, b_i, a_i).
$$
 Therefore, for the product $M_i =L_iL_{i+1} \cdots L_{i+n-1}$, we have
\begin{align}\label{mmtilde}
M_i = c_i\Lambda_{i-2}^{-1} \tilde M_{i-2}  D_{i-2}^{-1} \Lambda_{i+n-2},
\end{align}
where $\tilde M_j$ is given by \eqref{mimat}, and $c_i \in \C$. Further, multiplying \eqref{matDiffEqn} by the monodromy matrix $ \tilde M$ from the left and using \eqref{matMonEqn}, we get
$$
W_{i+n + 1}D_{i+1} = W_{i+n}D_iN_i.
$$
Comparing the latter equation with  \eqref{matDiffEqn}, we see that
\begin{align}
N_{i+n}  = D_iN_iD_{i+1}^{-1}.
\end{align}
Spelling out this equation, we get the relations
\begin{align}\label{quasiPeriod2}
t_{i+3} = t_i,\quad a_{i+n} = a_i \frac{t_{i+2}}{t_i}, \quad b_{i+n} = b_i \frac{t_{i+1}}{t_i}
\end{align}
(cf. Remark 4.4 of \cite{ovsienko2010pentagram}). This implies the following quasiperiodicity condition for $\Lambda_i$:
\begin{align}\label{quasiPeriod3}
\Lambda_{i+n} = \frac{1}{t_i}D_i\Lambda_i.
\end{align}
Using the latter equation, \eqref{mmtilde} can be  rewritten as
\begin{align*}
M_i = \frac{c_i}{t_{i-2}}\Lambda_{i-2}^{-1} \tilde M_{i-2} \Lambda_{i-2}.
\end{align*}
Since $\tilde M_{i-2}$ is conjugate to the monodromy matrix $M$, this proves the lemma.
\end{proof}
\begin{corollary}\label{scaledMonCor}
For an equivalence class $P \in \mathcal P_n$ with corner invariants $(x_i,y_i)$, the scaled monodromy is given by any of the matrices $M_i(z) := L_i(z)L_{i+1}(z)\cdots L_{i+n-1}(z),$
where
\begin{align}\label{canScaledLax}
L_{j}(z) := \left(\begin{array}{ccc}0 & 0 & 1 \\ -x_{j}y_{j} & 0 & 1 \\0 & -z^{-1}y_{j} & 1\end{array}\right). \end{align}
\end{corollary}
\begin{proof}
This follows from formula \eqref{rescal} for the rescaling action.
\end{proof}

\section{Monodromy invariants for inscribed polygons}\label{sec4}
In this section, we show that for inscribed polygons the scaled monodromy satisfies a certain self-duality relation and then use this to prove Theorem \ref{thm1}.
\begin{lemma}\label{lemma2}
Consider an equivalence class $P \in \mathcal P_n$ of polygons inscribed in a nondegenerate conic. Then the corresponding scaled monodromy matrix can be chosen to satisfy the self-duality relation
\begin{align}\label{sdRelation}
M(z) = (M(z^{-1})^{-1})^t.
\end{align}
\end{lemma}
\begin{remark}
We call \eqref{sdRelation} \textit{self-duality} because the matrix $(M(z^{-1})^{-1})^t$ represents the scaled monodromy of the dual polygon. \end{remark}
\begin{remark}
From \eqref{sdRelation} it follows, in particular, that the matrix $M(1)$, i.e., the actual non-scaled monodromy of the polygon, is orthogonal. This has a geometric explanation: if a polygon is inscribed into a conic~$C$, then the corresponding monodromy should preserve~$C$. But the subgroup of projective transformations preserving a non-singular conic is conjugate (over $\C$) to $\mathrm{SO}_3 \subset \P\mathrm{GL}_3$.
\end{remark}
\begin{proof}[Proof of Lemma \ref{lemma2}]

Take any polygon $(v_i)$ in the equivalence class $P$, and let $C$ be the conic circumscribed about $(v_i)$. Using the isomorphism $C \simeq \P^1$, we get a well-defined notion of cross-ratio of four points on $C$. Referring to this cross-ratio, we set 
$$
p_i := 1 - [v_{i-2}, v_{i-1}, v_i, v_{i+1}].
$$
The sequence $(p_i)$ is $n$-periodic and depends only on the projective equivalence class of the polygon $(v_i)$. By Lemma 3.1 of~\cite{schwartz2011pentagram}, we have the following relation between $p_i$'s and the corner invariants:
$$
x_i =\frac{1 - p_i}{p_{i+1}}, \quad y_i =\frac{1 - p_{i+1}}{p_i}.
$$
Using these formulas and Corollary \ref{scaledMonCor}, we express the scaled monodromy in terms of $p_i$'s:
\begin{align}\label{prodFormula}
 M_i(z) := L_ i(z)L_{i+1}(z)\cdots L_{i+n-1}(z),\quad
L_{j}(z) :=\left(\begin{array}{ccc}0 & 0 & 1 \\ {(1 - p_j^{-1})( p_{j+1}^{-1} - 1)} & 0 & 1 \\0 & z^{-1}p_j^{-1}({p_{j+1} - 1}) & 1\end{array}\right) .
\end{align}
Further, we notice that the matrix $L_j(z)$ is related to the matrix
\begin{align*}
L_{j}'(z) :=  \left(\begin{array}{ccc}0 & 0 & 1 \\ p_j-1 & 0 & p_j \\0 & {z^{-1}}({p_j - 1})^{-1} & {p_j}({1-p_j})^{-1}\end{array}\right) \end{align*}
by a gauge type transformation:
$$
L_{j}'(z) := \frac{p_j}{1-p_j}Q_j L_j Q_{j+1}^{-1}, \quad Q_j := \mathrm{diag}({p_j^{-1}-1}, 1 - p_j, 1).
$$
Therefore, the matrix $M_i(z)$ is conjugate to a scalar multiple of
$
 M_i'(z) :=L_i'(z)L_{i+1}'(z) \cdots L'_{i+n-1}(z).
$
Further, observe that the matrices $L'_j(z)$ satisfy the relation
$$
(L_j(z^{-1})^{-1})^t = S(z)^{-1}L_j'(z)S(z), \quad S(z) := \left(\begin{array}{ccc}1 & 1 & -1 \\1 & 1 & z \\-1 & z^{-1} & 1\end{array}\right).
$$
Therefore, the same relation is satisfied by their product $M_i'(z)$:
\begin{align}\label{twistedSD}
(M_i'(z^{-1})^{-1})^t = S(z)^{-1}M_i'(z)S(z).
\end{align}
Further, write the matrix $S(z)$ as
\begin{align}\label{factor}S(z) = T(z)T(z^{-1})^t,\quad T(z) := \left(\begin{array}{ccc}1 & 0& 0\\1 & \sqrt{-1} &1 \\-1&  -\sqrt{-1} & z^{-1}\end{array}\right),
\end{align}
and set
$M(z) := T(z)^{-1}M_1'(z)T(z)$. Note that since $M(z)$ is conjugate to $M_1'(z)$, it is also conjugate to a scalar multiple of $M_1(z)$, and thus can be taken as a scaled monodromy matrix for $P$.
Furthermore, from~\eqref{twistedSD} and~\eqref{factor} it follows that $M(z)$ satisfies~\eqref{sdRelation}, as desired.
\end{proof}
\begin{proof}[Proof of Theorem \ref{thm1}]
Consider a polygon $(v_i)$ inscribed in a nondegenerate conic. By Lemma~\ref{lemma2}, the corresponding equivalence class $P \in \mathcal P_n$ admits a scaled monodromy matrix with self-duality relation~\eqref{sdRelation}. From self-duality, it immediately follows that the corresponding normalized characteristic polynomial \eqref{specPolynomial} satisfies
$$
R(z^{-1},w^{-1}) = -z^nw^{-3}R(z,w)
$$
which means that the spectral curve $R(z,w) = 0$ is invariant under the involution $(z,w) \leftrightarrow (z^{-1}, w^{-1})$. In view of Corollary~\ref{mainCor}, this shows that $E_k = O_k$ for all $k$, proving the theorem.
\end{proof}

\medskip

\section{Discussion and open questions}\label{secLast}
\paragraph*{Circumscribed polygons.} By duality, Theorem \ref{thm1} is also true for circumscribed polygons. Furthermore, there is a ``circumscribed'' analog of Theorem \ref{thm2}. Namely, circumscribed polygons fill another torus parallel to the Prym variety (see Figure \ref{IC}). Note that, in general, the orbit of an inscribed polygon under the pentagram map does not have to intersect the circumscribed locus. However, there are some special cases when a certain iteration of an inscribed polygon under the pentagram map is circumscribed \cite{schwartz2010elementary}. It would be interesting to find an algebraic geometric explanation of this phenomenon.
 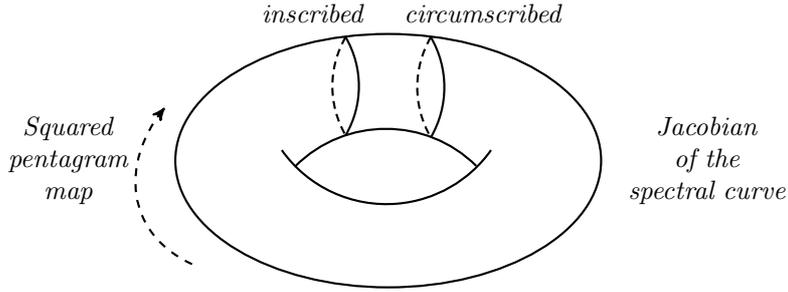
\begin{figure}[t]
\centerline{
\begin{tikzpicture}[thick, rotate = -90, scale = 1.4, every text node part/.style={align=center}]
    \draw (3,-3) ellipse (1.2cm and 2cm);
    \draw   (3.05,-3.87) arc (225:135:1.2cm);
    \draw   (2.9,-4) arc (-55:55:1.2cm);
    \draw  [densely dashed] (1.83,-3.4) arc (-120:-60:0.93cm);
        \draw   (1.83,-3.4) arc (120:60:0.93cm);
            \draw  [densely dashed] (1.83,-2.6) arc (-120:-60:0.95cm);
        \draw   (1.83,-2.6) arc (120:60:0.95cm);
        \node at (1.6,-3.7) () {\textit{inscribed}}; 
                \node at (1.6,-2.1) () {\textit{circumscribed}}; 
                    \draw [dashed, <-] (2.5,-5.1) .. controls (3, -5.5) and (3.7, -5.5)  .. (4,-4.8);
                              \node[text width=5cm] at (3,-6) () { \textit{Squared} \\ \textit{pentagram} \\ \textit{map}}; 
                                  \node[text width=5cm] at (3,0) (A) { \textit{Jacobian} \\ \textit{of the} \\ \textit{spectral curve}}; 
\end{tikzpicture}
}
\caption{Inscribed and circumscribed polygons.}\label{IC}
\end{figure}

\paragraph*{Continuous limit.} In the continuous limit, when the number of vertices of a polygon is large, while the vertices are close to each other, the pentagram map becomes the Boussinesq flow on curves \cite{ovsienko2010pentagram}. In the language of differential operators, the Boussinesq equation reads
$$
\partial_t L = [L, (L^{{2}/{3}})_+],
$$
where $L := \partial_x^3 + u(x) \partial_x + v(x),$
and $(L^{2/3})_+$ denotes the differential part of the pseudodifferential operator $L^{2/3}$. The continuous limit of both inscribed and circumscribed polygons are parametrized conics, corresponding to skew-adjoint operators~$L$. This, in particular, means that in the limit the inscribed and circumscribed tori in Figure \ref{IC} collide.\par
Similarly to how the pentagram map does not preserve inscribed polygons, the Boussinesq flow does not preserve the set of skew-adjoint operators $L$. However, odd flows of the Boussinesq hierarchy, that is, flows of the form
$$
\partial_t L = [L, (L^{{m}/{3}})_+],
$$
where $m \in \Z$ is odd and not divisible by $3$, do preserve skew-adjoint operators. Is there a discrete analog of this statement? Does there exist an explicit map commuting with the pentagram map and preserving the set of inscribed polygons?

\paragraph*{Poisson structure.} By Theorem \ref{thm1}, the moduli space of inscribed polygons is foliated by (open dense subsets of) Prym varieties. Does there exist a Poisson structure making this foliation into an algebraic completely integrable system? (Note that the Poisson structure of the pentagram map does not restrict to inscribed polygons.) \par
One possible approach to this problem is to use the fact that matrices satisfying self-duality relation~\eqref{sdRelation} form a subgroup in the loop group of $\P\mathrm{GL}_3$, somewhat similar to a twisted loop group. Does this subgroup admit a Poisson-Lie structure? If the answer to the latter question is positive, one should be able to use this structure to construct a Poisson bracket on inscribed polygons.

\paragraph*{Self-dual polygons.} Another class of polygons satisfying $E_k = O_k$ are {self-dual} polygons. Recall that, for a polygon $(v_i) \in \P^2$, the vertices of the \textit{dual polygon} are sides $[v_i, v_{i+1}]$ of the initial one, regarded as points in the dual projective plane. A polygon $(v_i) \in \P^2$ is called \textit{self-dual} if there exists some $m \in \Z$ and a projective map $  \P^2 \to (\P^2)^*$ taking the vertex $v_i$ to the vertex $[v_{i+m}, v_{i+m+1}]$ of the dual polygon. Furthermore, for polygons with odd number of vertices, there exists a canonical notion of self-duality. We say that a $(2q+1)$-gon $(v_i) \in \P^2$ is  \textit{canonically self-dual} if there exists a projective map $  \P^2 \to (\P^2)^*$ taking the vertex $v_i$ to the vertex $[v_{i+q}, v_{i+q+1}]$ of the dual polygon (i.e., to the opposite side). Canonically self-dual polygons form a maximal dimensional stratum in the set of self-dual polygons \cite{fuchs2009self}. Furthermore, one can show that canonically self-dual polygons correspond to a third torus in Figure \ref{IC} parallel to the inscribed and circumscribed tori. 
\par
This description suggests that there should exist a (possibly birational) isomorphism between the varieties of canonically self-dual and inscribed polygons. For closed $7$-gons and $9$-gons, such an isomorphism is constructed in \cite{schwartz2010elementary}. In particular, for $7$-gons it is given by the pentagram map. Are the varieties of canonically self-dual and inscribed polygons isomorphic in the general case?

\paragraph*{Poncelet polygons.}
For certain singular spectral curves, the inscribed and circumscribed tori in Figure~\ref{IC} coincide, giving rise to polygons which are simultaneously inscribed and circumscribed. Such polygons are known as \textit{Poncelet polygons}, since they are closely related to the famous Poncelet porism. In \cite{schwartz2007poncelet} it is proved that Poncelet polygons are fixed points for the squared pentagram map. Is there an algebraic geometric interpretation of this statement? Does the squared pentagram map have other fixed points? Are the fixed points given by Poncelet polygons Lyapunov stable?

\paragraph*{Degenerate conics.} How do the results of the present paper generalize to polygons inscribed in degenerate conics? Consider, in particular, a $2q$-gon with the following property: all its odd vertices lie on a straight line $l_1$, while all even vertices lie a on a line $l_2$. Such polygons exhibit an interesting behavior under the (inverse) pentagram map, known as the Devron property \cite{glick2015devron}. Is there an algebraic geometric explanation for the Devron property? Note that the corresponding spectral curves are highly singular and have the form
$$
(w - a^{-1}b^{-1})(w - az^{-q})(w - bz^{-q}) = 0
$$
for certain $a, b \in \C$. We believe that it should it be possible to relate the Devron property with algebraic geometry of such singular curves.

\medskip

\section{Appendix: Positivity of the monodromy invariants for convex polygons}
Following \cite{schwartz2011pentagram}, we define the following ``signed'' versions of the monodromy invariants:
$$
O^*_k := (-1)^kO_k, \quad E^*_k := (-1)^kE_k, \quad k = 1, \dots, \left[\frac{n}{2}\right],
$$
and also set
$
O^*_n := O_n$, $E^*_n := E_n.
$
The aim of this appendix is to prove the following result conjectured in~\cite{schwartz2011pentagram}.
\begin{theorem}\label{thm:positivity}
For convex closed polygons, one has $E_k^* > 0$, $O_k^* > 0$ for every $k$.
\end{theorem}
The proof is based on the following lemma.
\begin{lemma}\label{lemma:pos}
\begin{enumerate}
\item The signed monodromy invariants $E_k^*$ can be written as polynomials with positive coefficients in terms of $y_i, 1-x_iy_i$.
\item The signed monodromy invariants $O_k^*$ can be written as polynomials with positive coefficients in terms of $x_i, 1-x_iy_i$. 
\end{enumerate}
\end{lemma}
\begin{proof}
We begin with the first statement. Since the result obviously holds for $E_n^* = E_n$, we only consider the case $k \leq [n / 2]$. 
Using the matrix $M_i(z)$ from Corollary \ref{scaledMonCor} and the definition of the monodromy invariants, we get
$$
1+\sum_{k=1}^{[n/2]}E^*_k(-z)^{-k} =\, \mathrm{trace}\, M_i(z).
$$
Therefore, to prove that $E_k^*$ is a polynomial with positive coefficients in terms of $y_i, 1-x_iy_i$, it suffices to show that  $\mathrm{trace}\, M_i(z)$ is a polynomial with positive coefficients in $y_i$, $1-x_iy_i$, and $-z^{-1}$. To that end, notice that the matrix $L_j(z)$ from Corollary \ref{scaledMonCor} can be written as 
$$
L_j(z) = U\hat L_j(z)U^{-1}, \quad
 \hat L_j(z) := \left(\begin{array}{ccc}1 & 0 & 1 \\1-x_jy_j & 0 & 1 \\0 & -z^{-1}y_j & 0\end{array}\right), \quad U := \left(\begin{array}{ccc}1 & 0 & 0 \\0 & 1 & 0 \\1 & 0 & 1\end{array}\right).
$$
Therefore, the product $M_i(z)=  L_i(z)\cdots  L_{i+n-1}(z)$ satisfies the same relation
$$
M_i(z) = U \hat M_i(z)U^{-1}, \quad \hat M_i(z) := \hat L_i(z)\cdots \hat L_{i+n-1}(z).
$$
To complete the proof, notice that since all non-zero entries of the matrix $ \hat L_j(z)$ are polynomials with positive coefficients in terms of $y_i$, $1-x_iy_i$, and $-z^{-1}$, the same is true for the entries of the matrix $ \hat M_i(z)$, and thus $\mathrm{trace}\, \hat M_i(z) = \mathrm{trace}\,  M_i(z)$, as desired.\par
Now, we prove the second statement. Let 
$$
L_j^*(z) := \frac{z}{y_j}\mathrm{adj}(L_j(z)) = \left(\begin{array}{ccc}1 & -1 & 0 \\x_j z & 0 & -x_j z \\x_jy_j & 0 & 0\end{array}\right),
$$
where $\mathrm{adj}$ stands for the adjoint, i.e., the transposed cofactor matrix, and $L_j(z)$ is the matrix from Corollary~\ref{scaledMonCor}. Then we have$$
1+\sum_{k=1}^{[n/2]}O^*_k(-z)^{k} = \mathrm{trace}\, M^*_i(z), \quad  M_i^*(z) := L_i^*(z)\cdots  L^*_{i+n-1}(z).
$$
Now, we notice that $L_j^*(z)$ can be written as
$$
L_j^*(z) = V\hat L_j^*(z) V^{-1}, \quad
 \hat L_j^*(z) := \left(\begin{array}{ccc}0& 1 & 1 - x_jy_j \\-x_jz & 0 & 0 \\0 & 1 & 1\end{array}\right), \quad V := \left(\begin{array}{ccc}0 & 0 & 1 \\0 & -1 & 0 \\-1 & 0 & 1\end{array}\right),
$$
and then apply the same argument as in the proof of the first statement. Thus the lemma is proved.
\end{proof}
\begin{remark}
As was observed by R.\,Schwartz, Lemma~\ref{lemma:pos} can also be proved using the concepts of \textit{right and left modifications} introduced in Section 2.2 of \cite{schwartz2008discrete}.
\end{remark}
\begin{proof}[Proof of Theorem \ref{thm:positivity}]
It is easy to see that for convex closed polygons the corner invariants satisfy the inequalities $0 < x_i, y_i < 1.$
So, for such polygons one has $x_i, y_i, 1-x_iy_i > 0$, and the result of the theorem follows from Lemma~\ref{lemma:pos}.
\end{proof}

\bibliographystyle{plain}
\bibliography{oe.bib}

\end{document}